\documentclass[journal]{IEEEtran}

\usepackage{multicol,theorem,amsmath,array,cmmib57}
\usepackage{amssymb}
\usepackage{latexsym}
\usepackage{amsfonts}
\usepackage{bbm}

\newtheorem{thm}{Theorem}
\newtheorem{prop}[thm]{Proposition}
\newtheorem{deff}{Definition}

\newtheorem{cor}[thm]{Corollary}

\newcommand{\F}{\mathbb{F}}

\newcommand{\wt}{\mbox{\rm wt}}

\newcommand{\dd}{d^{\perp} }

\title{Results on Binary Linear Codes With Minimum Distance $8$ and $10$}
\author{Iliya Bouyukliev  and Erik Jacobsson
  \thanks{Iliya Bouyukliev is with the Institute of Mathematics and Informatics, Bulgarian Academy of Sciences, P.O.Box 323, 5000 Veliko Tarnovo,
Bulgaria, iliya@moi.math.bas.bg}
\thanks{Erik Jacobsson is with the
Department of Mathematical Sciences, University of Gothenburg and Chalmers University of Technology, S-412 96 Gothenburg, Sweden,
erik.jakobsson@chalmers.se}
}
\date{\today}

\begin{document}
\maketitle

\begin{abstract}
All codes with minimum distance $8$ and codimension up to $14$ and
all codes with minimum distance $10$ and codimension up to $18$
are classified. Nonexistence of codes with parameters [33,18,8]
and [33,14,10] is proved. This leads to 8 new exact bounds for
binary linear codes. Primarily two algorithms considering the dual
codes are used, namely extension of dual codes with a proper
coordinate, and a fast algorithm for finding a maximum clique in a
graph, which is modified to find a maximum set of vectors with the
right dependency structure.
\end{abstract}


\section{Introduction}

Let $\F_2^n$ denote the $n$-dimensional vector space over the
field $\F_2$ and let the inner product
 $\left< \cdot ,\cdot\right>: \F_2^n \times \F_2^n \rightarrow \F_2$ be defined in the natural way as
 $\left<u,v\right> = \sum_{i=1}^n u_i v_i$, where addition is in $\F_2$. The Hamming
distance between two vectors of $\F_2^n$ is defined as the number
of coordinates in which they differ, and the weight $\wt(v)$ of a
vector $v \in \F_2^n$ is the number of the nonzero coordinates of
$v$. A
linear binary $[n,k,d]$ code $C$ is a $k$-dimensional subspace of
$\F_2^n$ with minimum distance $d = \min\{\wt(c) :
 c \in C, c\neq 0 \}$. A generator matrix $G$ for an
$[n,k]=[n,k,d \geq 1]$ code is any matrix whose rows form a basis
for the code. The orthogonal complement $C^\perp$ of $C$ in
$\F_2^n$ is called the dual code of $C$ and is an
$[n,n-k,d^\perp]$ code, where $d^\perp$ is called the dual
distance of $C$ and $n-k$ the codimension. A generator matrix $H$
for the dual code is called a parity check matrix of the code $C$.
In this paper we will say that a code $C$ is an
$[n,k,d]^{d^\perp}$ code if it is an $[n,k,d]$ code with dual
distance $d^\perp$. Further, two binary linear codes $C_1$ and
$C_2$ are said to be equivalent if there is a permutation of
coordinates which sends $C_1$ to $C_2$.
 Throughout this paper all codes are assumed to be binary.

A central problem in coding theory is that of optimizing one of
the parameters $n$, $k$ and $d$ for given values of the other two.
Usually this optimization is related to the following functions:
$n_2(k,d)$ - the minimum length of linear codes for given minimum
distance $d$ and dimension $k$ and $d_2(n,k)$  the largest value
of $d$ for which a binary [$n,k,d$] code exists.  Codes with
parameters $[n_2(k,d),k,d]$ and $[n,k,d_2(n,k)]$ are called
\textit{optimal}. There are many reasons to study optimal codes.
These codes are interesting not only for detection and correction
of errors. Some of them have rich algebraic and combinatorial
structure. The problems of optimality are strongly connected and
can be considered as packing problem in statistics and in finite
projective spaces \cite{HirSt}. Unfortunately, all these problems
are, as many others in coding theory, computationally difficult
\cite{Barg} and the exact values of $n_2(k,d)$ for all $d$ are
known only for $k\le 8$ \cite{BJV_IEEE00}. Tables with bounds and
exact values for $d_2(n,k)$  are given in \cite{Brou98} and
\cite{new-Tables}.

Another application of optimal codes is directly related to the
design method of cryptographic Boolean functions suggested by
Kurosawa and Satoh \cite{Kurosawa}. In this case the optimal
linear codes have to be with largest possible dual distance. More
precisely one have to study  the function $N(d,d^\perp)$ as the
minimal $n$ such that there exists a linear binary code of length
$n$ with minimum distance $d$ and dual distance $d^\perp$. The
investigation of $N(d,d^\perp)$ seems to be much harder than the
investigation of $n_2(k,d)$. There are some general bounds (see
\cite{Matsumoto}) but these bounds can be reached only for a few
values of $d$ and $d^\perp$. In a previous work, we studied
$N(d,d^\perp)$ for $d^\perp\leq d \le 12$ by computer using the
package \textsc{Q-Extension} \cite{Bo_SJC07}. With this package we
attempted to construct generator matrices and classify codes with
fixed parameters. Practically, we had no success in the cases
where the dual distance was more than 6. In the case of codes with
fixed minimum distance larger than 2, it is quite natural to look
at the duals of the codes with needed  properties and to the
parity check matrices. In other words, extending an $[n-1,k-1,d]$
code $C_1$ to an $[n,k,d]$ code $C_2$ can be considered as
extending an $[n-1,n-k]^{d}$ to an $[n,n-k]^{d}$ code with one
coordinate. This approach helps us to develop two different
methods, which are much more effective when we study
$N(d,d^\perp)$, and also $n_2(k,d)$, for $d= 8$ and $d=10$. For
small dimensions, it is convenient to use a brute force algorithm
that takes generator matrices of all inequivalent $[n,k,d]^{\dd}$
codes as input, extends them in all possible ways and then checks
the constructed codes for minimum and dual distance and for
equivalence. As a result we get all inequivalent $[n+1,k,\ge
d]^{\dd}$ codes. This method is of course impossible to use for
larger dimensions $k$, because of the large number of possible
extensions. To avoid this problem we use a second method for
larger dimensions adopting a strategy for bounding of the search
space similar to the strategy for finding a maximum clique in a
graph suggested in \cite{Ostergard-clique}.

 In this paper we present two algorithms which can be used for
 constructing of linear codes with fixed dual distance. We give
 classification results for  all codes with minimum
 distance 8 and codimension up to 14 and minimum
 distance 10 and codimension up to 18.
We would like to refer to \cite{Jaff97} and \cite{Jaffe} for a detailed bibliography of works which study linear codes with minimum
distance 8 and 10.

\section{Preliminaries}

In this Section we give some properties of linear codes and the
relations with their dual codes which help us to design the
construction algorithms.

The first proposition considers the even linear codes, namely the
linear binary codes which consist only of even weight vectors.

\begin{prop}
If $d\geq 2$ is even and a linear $[n,k,d]$ code exists, then
there exists an even $[n,k,d]$ code.
\end{prop}

\begin{proof}
If $C$ is a linear $[n,k,d]$ code, the code produced by puncturing $C$ in one coordinate has parameters $[n-1,k,d$ or $d-1]$. Adding a
parity check bit to all codewords, we obtain an even $[n,k,d]$
code.
\end{proof}

Practically, we use the following corollary:

\begin{cor}
If $d\geq 2$ is even and even linear $[n,k,d]$ codes do not exist,
then no $[n,k,d]$ code exists.
\end{cor}

Later on we give the definition and some properties of residual
codes.

\begin{deff}\label{res_def}
The residual code $Res(C,c)$ with respect to a codeword $c \in C$
is the restriction of $C$ to the zero coordinates of $c$.
\end{deff}

A lower bound on the minimum distance of the residual code is
given by
\begin{thm}(\cite{Hand},Lemma 3.9) \label{res_d}
Suppose $C$ is a binary $[n,k,d]$ code and suppose $c \in C$ has
weight $w$, where $d > w/2$. Then $Res(C,c)$ is an $[n-w,k-1,d']$
code with $d' \ge d - w + \lceil{w/2}\rceil$,
\end{thm}
and on the dual distance by
\begin{prop}\label{res_dd} Suppose $C$ is a
binary $[n,k,d]$ code with dual distance $d^{\perp}$, $c \in C$,
and the dimension of $Res(C,c)$ is $k-1$. Then the dual distance
of $Res(C,c)$ is at least $d^{\perp}$.
\end{prop}

There is also a well known elementary relationship between the
minimum distance of a linear code and the parity check matrix.
\begin{prop}\label{pcm}
A linear code has minimum distance $d$ if and only if its parity
check matrix has $d$ linearly dependent columns but no set of
$d-1$ linearly dependent columns.
\end{prop}

The next proposition gives a connection between weights of the
rows of generator matrices and columns of parity check matrices.
\begin{prop}\label{even}
Any even linear code $C$ has a parity check matrix whose columns
have odd weights.
\end{prop}
\begin{proof}
Let $G = \left[ I_k P \right]$ be a generator matrix for $C$ in
standard form. Then every row of $P$ has odd weight. And
accordingly the parity check matrix $H= \left[ P^T I_{n-k}
\right]$ has only odd weight columns. The sum of all rows of $H$
gives the all-ones vector.
\end{proof}

\begin{prop}\label{weight1}
If a linear code $C$ has $t$ codewords of weight $1$, then its
parity check matrix has $t$ zero-columns.
\end{prop}
\begin{proof}
Without loss of generality, let $u=(1,0, \hdots ,0) \in C$. Then
$\forall v \in C^\perp$, $0 = \left< u,v \right> = \sum_{i=1}^n
u_i v_i = v_1$.
\end{proof}
This means that if we have found all inequivalent codes with dual
distance $8$ and dimension $14$, then we have also
found all codes with minimum distance $8$ and codimension up to $14$.\\
Another well known fact is that for $d$ odd,
\begin{align*}
|\left\{ \text{inequivalent } [n,k,d]  \text{ codes} \right\}| \\
 \ge |\left\{\text{inequivalent even }[n+1,k,d+1] \text{ codes} \right\}|
\end{align*}
with equality only if the automorphism groups are transitive.

We define for convenience the function $L(k,d^\perp)$.
\begin{deff}
Let $d^\perp \ge 3$. Then $L(k,d^\perp)$ is  the maximum length
$n$  such that a binary $[n,k]^{\dd}$ code exists.
\end{deff}

The next theorem holds.
\begin{thm}\label{b_th}
Let $C$ be a binary linear $[n,k,d]^{\dd}$ code. Then $d \ge n -
L(k-1,d^\perp)$.
\end{thm}
\begin{proof}
Let $C$ be a code with parameters $[n,k,d]^{\dd}$ and $c_d$ be a
codeword in $C$ of weight $d$. If $d < n - L(k-1,d^\perp)$ then
the residual code $Res(C,c_d)$ has dual distance at least
$d^\perp$, length $n-d > L(k-1,d^\perp)$ and dimension $k-1$. This
is impossible because $L(k,d_1^\perp) \ge L(k,d_2^\perp)$ for
$d_1^\perp \ge d_2^\perp$.
\end{proof}

In fact, the function $L(k,d^\perp)$ has already been investigated
in another setting for dual distance greater than or equal to $4$
in connection with studies on $\kappa$-caps in projective
geometries since it is known that a $\kappa$-cap in $PG(k-1,q)$ is
equivalent to a projective $q$-ary $[n=\kappa,k]^{d^\perp}$ code
with $d^\perp \ge 4$. And $L(k,d^\perp)$ have thus been considered
in connection with the function
\begin{align*}
\mu_\nu(N,q) = & \text{ the maximum value of } \kappa \text{ such that there exist} \\
&\text{ a $\kappa$-cap in } PG(N,q),
\end{align*}
where $q$ is the order of the underlying Galois field, in our case $2$.
 More information on this can be found in the survey by Hirschfeld and Storme \cite{HirSt}.

\section{Computational  Algorithms}

We use mainly two algorithms for the extension of codes.
If the dimension of the considered codes with given dual distance
is small we can find the next column of the generator matrix of
the new code relatively easy because we can represent any
generator matrix $G$ in the packing form as a vector $G_b$ of $n$ computer words.
This algorithm, named \textsc{Bruteforce}, can be described with the following steps:\\

Algorithm \textsc{Bruteforce}:\\

 \noindent
    INPUT: $ C_{inp}$ - Set of all inequivalent $[n,k]^{d^\perp}$ codes represented by their generator matrices in packing form.\\
    OUTPUT: $C_{out}$ - Set of all inequivalent $[n,k+1]^{d^\perp}$
    codes.\\
    var a:array$[1..2^k-1]$ of  integer;\\

    {In the beginning, $C_{out}$ is the empty set.
    For any code $C_r$ in $C_{inp}$ with generator matrix in
    packing form $G_r$ do the following:}
\begin{enumerate}

\begin{item}{    Set a[i]:=1 for any i.}
\end{item}
\begin{item}{
    Find all linear combinations $b$ of up to $d^\perp-2$ column vectors of $G_r$ and set a[b]:=0.}
\end{item}
\begin{item}{
    For all $j$ such that $a[j]=1$ extend $G_r$ with one coordinate, equal to $j$, to $G_r'$. If
    there are no codes in $C_{out}$ equivalent to $C_r'$ (generated by $G_r'$) do $C_{out}:=C_{out}\cup C_r'$}
\end{item}
    \end{enumerate}

\noindent  The big advantage of this algorithm is given by Step 2.
In that step, all possible solutions for the $(n+1)$th column of
the generator matrices are determined with approximately with
$\sum_{i=1}^{d-2}{n \choose i}$ operations. Actually, to find all
vector solutions for the $(n+1)$th column, we take all
$k$-dimensional vectors and delete those which are not
solutions. We find all sums  of less than $\dd-2$ columns of the
known part of the generator matrix. Each sum gives us one vector
which is not a solution and have to be deleted. All remaining
vectors are solutions.

 In Step 3, we use canonical representation
of the objects. The main priority of the canonical representation
is that the equivalence (isomorphism) test is reduced to check of
coincidence of the canonical representations of the structures. In
the case of many inequivalent codes,  the computational time for
comparing is growing fast. A technique for surmounting  this problem
is worked out. We split the set of inequivalent codes into a big
amount of cells according to a proper invariant.

To explain the next algorithm, we need the following definition.

 \begin{deff}\label{p-proper}Let $E$ be a set of $k$-dimensional vectors.
    \begin{enumerate}
    \item We call $E$ \textit{p-proper} if all subsets of  $p$ vectors of $E$ are linearly independent.
    \item Let $M$ be a $k \times n$ matrix. The set $E$ is called $p$-proper with respect to $M$ if $E \cup \{\text{columns in } M \}$ is a $p$-$proper$ set.
    \end{enumerate}
\end{deff}
\noindent Observe that, by Proposition \ref{pcm}, the columns of a parity check matrix for an $[n,k,d]$ code form a $(d-1)$-proper set.\\

 We consider the following problem: How to find a set of  $t=d-1$ binary vectors which have a certain property, i.e.,
     $(\dd-1)$-proper  subset  of the set of all possible binary vectors with respect to a fixed generator matrix.
    To attempt to solve this problem in reasonable time, we adopt an idea suggested by \"Osterg{\aa}rd in
    \cite{Ostergard-clique}
     for finding a maximum clique in a graph in the algorithm \textsc{Extend}.

   Let $C$ be an $[n,k,d]^{\dd}$ code with a generator matrix $G$ in the form
    \begin{equation} G =
        \left[
        \begin{array}{cc|c}
        0 0 \hdots 0 & 1 & 1 1 \hdots 1 \\ 
         & 0 &   \\
        Res_d(C)  &  \vdots &  X  \\
          &  0      &
        \end{array}
        \right]
        =
        \left[
        \begin{array}{c|c}
        A  &  \hat{X}
        \end{array}
        \right]
        \end{equation}
where $Res_d(C)$ is a generator matrix of the residual
$[n-d,k-1,\geq \frac{d}{2}]^{\ge\dd}$ code. Given that we know all
such inequivalent generator matrices the problem is reduced to
finding all $(\dd-1)$-proper sets $\hat{X}$ with respect to $A$ of
$d-1$ binary vectors on the form $(1,x_2,\hdots,x_k)^T$ .

Let
\begin{equation*} V^* = \left\{ (1,x_2,x_3,\hdots,x_k): x_i \in \F_2 \right\}
\end{equation*}

(\textit{Remark}: If the dual code is even we may, by Proposition \ref{even}, reduce the search space to the set of odd-weight binary vectors).
Delete from $V^*$ all linear combinations of $\dd-2$, or less, vectors from $A$. The remaining set
\begin{equation}
V = \left\{ v_1, v_2, \hdots, v_N \right\}
\end{equation}
 is the search space for our search strategy.
Now, for each integer $1 \leq i \leq N$, let

\begin{equation}
V_i = \left\{ v_i, v_{i+1}, \hdots, v_N \right\}
\end{equation}
and let $\textbf{r}$ be the $N$-tuple, defined by $\textbf{r}[i]=
\min\{s,t \}$, $1\leq i \leq N$, where $s$ is the size of the
largest $(\dd-1)$-proper subset of $V_i$. First we consider
$(\dd-1)$-proper subsets with respect to $A$ of $V_N$ that contain
the vector $v_N$, this obviously is $\left\{ v_N \right\}$, and we
record the size of the largest proper subset found up to now in
the tuple ${\bf r}$, so $\textbf{r}[N]=1$.

In the $i$-th step we consider $(\dd-1)$-proper subsets with
respect to $A$ in $V_i$ containing $v_i$ and record the minimum
between  size of the largest proper subset found up to now and $t$
in ${\bf r}[i]$ (row *** in the algorithm).

The tuple ${\bf r}$ for the already calculated steps enables the
pruning strategy for the search. Since we are looking for a proper
subset of size $t = d-1$, and if the vector $v_i$ is to be the
$(size)$st vector
 in the subset and $size+{\bf r}[i] < t$, then we can prune the search (row * in the algorithm).
 When the search terminates, the size $s$ of the largest $(\dd-1)$-proper
 subset  with respect to $A$ of $V$ or $t$ (if $t<s$) will then be recorded in ${\bf r}[1]$.

In Step $size$, $size>1$, we choose all $k$ dimensional vectors
with first coordinate 1 which are not linearly dependent with
$\dd-2$ column vectors of the constructed until now part of the
generator matrix. Our idea is with one pass to find all proper
vectors (all elements of $U_{size}$) using all column vectors from
 the generator matrix obtained until this step using $U_{size-1}$ (all proper
 vectors  from previous step). To find all proper
 vectors $U_{size}$ we take  $U_{size-1}$
 and delete those which are not proper, with respect to the
already constructed part (row ** in the algorithm). We find all
sums  of less than $\dd-2$ columns. Each sum gives us one vector
which is not $(\dd-1)$-proper to the current step and have to be
deleted. All remaining vectors are proper. To improve the speed of
the algorithm, we pack each column in a computer word and use the
bit operation XOR for computer words. The presented algorithm is
much faster than the algorithm in \cite{BV_IEEE05}.

\section{Results}

In this section, we present obtained result for codes with dual
distance 8 and 10 using the above algorithms. With algorithm
\textsc{Bruteforce} we construct  all codes with length $n\leq
28$, dimension $k\le 14$ and dual distance at least 8, and all
codes with dimension up to 18 and dual distance at least 10. The
summarized results for the number of inequivalent codes for given
parameters are presented in the tables below. The stars in some
cells mean that for the corresponding parameters $n$ and $k$ there
are codes with dual distance greater than the considered one. In
the remaining cases,  the number of inequivalent codes of length
$n$ and dimension $k$, given in the table,  coincide with the
number of optimal codes with minimum distance 8 (respectively 10)
which have dimension $n-k$ and length not larger than $n$. We can
use the numbers in the tables to determine the exact number of
inequivalent optimal codes with length $n'=n$ and dimension
$k'=n-k$ in some of the cases. For the cells without $*$, the
number of inequivalent optimal $[n'=n,k'=n-k,8]$ codes is equal to
the number of inequivalent $[n,k]^d$ codes minus the number of
inequivalent $[n-1,k-1]^d$ codes. The calculations took about 72
hours in contemporary PC. The number of all inequivalent codes
with dimension 15 and dual distance 8, and dimension 19 and dual
distance 9, grows exponentially, so we could not calculate all
cases. That is why we consider the problems for existence of codes
with parameters $[33,18,8]$ and $[33,15,10]$.

The existence of a $[33,18,8]$ code leads to the existence of a
$[33,18,8]$ even code and $[32,17,8]$ even code (from the
properties of shortened codes and Lemma \ref{even}) and its dual
code $C_{32}$  with parameters $[32,15,d]^8$. We know that
$L(14,8)=28$ (see Table 1) and $d \le 8$ from the tables for
bounds of linear codes \cite{new-Tables}. Theorem \ref{b_th} gives
us that the minimum distance $d$ of $C_{32}$ has to be $4 \le d
\le 8$. Using the algorithm \textsc{Extend} and already
constructed even codes with parameters $[28,14,8]^8$,
$[27,14,7]^8, \dots, [24,14,d \ge 4]^8$, we obtain that there are
exactly two inequivalent even codes $C_{32}^1$ and $C_{32}^2$ with
generator matrices $G_{32}^1$ and $G_{32}^2$ and weight
enumerators:

\noindent
$1+124z^{8}+1152z^{10}+3584z^{12}+6016z^{14}+11014z^{16}+6016z^{18}+3584z^{20}+1152z^{22}+124z^{24}+z^{32}$, and\\
$1+116z^{8}+1216z^{10}+3360z^{12}+6464z^{14}+10454z^{16}+6464z^{18}+3360z^{20}+1216z^{22}+116z^{24}+z^{32}.$\\

\begin{table}[htb]
 {\small \[
 G_{32}^1=\left( \begin{array}{c}
10000000000000000000000001111111\\
01000000000000001010101101100011\\
00100000000000000011100001010101\\
00010000000000000011010001011010\\
00001000000000011000101011110100\\
00000100000000011011101011111011\\
00000010000000010111111110111111\\
00000001000000010111110010110000\\
00000000100000010111110100001100\\
00000000010000011010010100101110\\
00000000001000001101100010101101\\
00000000000100011110101101010000\\
00000000000010000110011010011101\\
00000000000001001000110110010111\\
00000000000000100110101100111001\\
  \end{array}
\right)
\]}
\end{table}

\begin{table}[htb]
 {\small \[ G_{32}^2=
\left( \begin{array}{c}
10000000000000000000000001111111\\
01000000000000001010101101110100\\
00100000000000000011100001010101\\
00010000000000011011100111001110\\
00001000000000010111101011111000\\
00000100000000001100011101100011\\
00000010000000011000111110100100\\
00000001000000011000110010101011\\
00000000100000011000110100010111\\
00000000010000010101010100101110\\
00000000001000001101100010110110\\
00000000000100010001101101100011\\
00000000000010000111110110110000\\
00000000000001000110011010001001\\
00000000000000100110101100101110\\
  \end{array}
\right)
\]}
\end{table}

None of these two codes can be extended to a code with parameters
$[33,15]^8$. This leads to

\begin{thm}
Codes with parameters $[33,18,8]$ do not exist and $n_2(18,8)=34$.
\end{thm}

It follows that codes with parameters $[32,18,7]$ and $[33,19,7]$
do not exist and the minimum distance of a putative
$[33,19,d]^{10}$ code has to have $d<7$. From Theorem \ref{b_th} and
value of $L(18,10)=28$ we have that $d$ has to be 5 or 6.


 There are exactly 30481 codes with parameters $[27,18]^{10}$ and
 11 codes with parameters $[26,18]^{10}$ and the dual codes of all  these codes are even.
But none can be extended to a code with parameters
$[33,19,d]^{10}$.
 This leads us to the conclusion

\begin{thm}
Codes with parameters $[33,14,10]$ do not exist and $n_2(14,10)=34$.
\end{thm}

The calculations  for nonexistence of codes with parameters
$[33,18,8]$ and $[33,14,10]$ took about 4 weeks in a contemporary
PC.

From Theorem 9 and Theorem 10  and tables for bounds of codes
\cite{new-Tables}, we have:

\begin{cor}
$n_2(19,8)=35$, $n_2(20,8)=36$, $n_2(21,8)=37$, $n_2(22,8)=38$,
$n_2(15,10)=35$ and  $n_2(16,10)=36$.
\end{cor}

\begin{table}
\begin{center}
{\bf TABLE 1 - Classification results for $[n,k]^{d^\perp \ge 8}$ codes} \\

\bigskip

\begin{tabular}{| c | c | c | c | c | c |}\hline
$\bf{n} \setminus \bf{k}$ & \bf{14} & \bf{13} & \bf{12} & \bf{11} & \bf{10} \\
\hline \hline
10 & 0 & 0 & 0 & 0 & 1* \\
\hline
11 & 0 & 0 & 0 & 1* & 4* \\
\hline
12 & 0 & 0 & 1* & 5* & 1 \\
\hline
13 & 0 & 1* & 6* & 3 & 0 \\
\hline
14 & 1* & 7* & 7* & 1 & 0 \\
\hline
15 & 8* & 14* & 4 & 1 & 0 \\
\hline
16 & 24* & 16 & 5 & 1 & 0 \\
\hline
17 & 50* & 23  & 5 & 0 & 0 \\
\hline
18 & 131 & 39 & 2 & 0 & 0 \\
\hline
19 & 450 & 30 & 1 & 0 & 0 \\
\hline
20 & 1863 & 27 & 1 & 0 & 0 \\
\hline
21 & 11497 & 13 & 1 & 0 & 0 \\
\hline
22 & 46701 & 10 & 1 & 0 & 0 \\
\hline
23 & 40289 &  9& 1 & 0 & 0 \\
\hline
24 & 5177 & 10 & 1 & 0 & 0 \\
\hline
25 & 536 & 8 & 0 & 0 & 0 \\
\hline
26 & 274 & 0 & 0 & 0 & 0 \\
\hline
27 & 1 & 0 & 0 & 0 & 0 \\
\hline
28 & 1 & 0 & 0 & 0 & 0 \\
\hline
29 & 0 & 0 & 0 & 0 & 0 \\

\hline
\end{tabular}

\end{center}
\end{table}

\begin{table}
\begin{center}

{\bf TABLE 2 - Classification results for $[n,k]^{d^\perp \ge 10}$ codes} \\

\bigskip

\begin{tabular}{| c | c | c | c | c | c |}\hline
$\bf{n} \setminus \bf{k}$  & \bf{18} & \bf{17} & \bf{16} & \bf{15} & \bf{14} \\
\hline \hline
15 & 0 & 0 & 0 & 0 & 1* \\
\hline
16 &  0 & 0 & 0 & 1* &6* \\
\hline
16 &  0 & 0 & 1* & 7* &3 \\
\hline
17 &  0 & 1* & 8* & 7& 0\\
\hline
18 &  1* & 9* & 14 & 1 &0\\
\hline
19 &  10* & 24* & 7 & 0 &0\\
\hline
20 &  38* & 29* & 3 &0 & 0\\
\hline
21 &  90* & 30& 2 &0 & 0 \\
\hline
22 &  237* &39 & 0 &0 &0\\
\hline
23 &  1031* & 29& 0 &0 &0\\
\hline
24 &  11114 &6 &0 &0 &0 \\
\hline
25 & 188572  & 0 & 0 &0 &0\\
\hline
26 &  563960  &0 & 0 &0 &0\\
\hline
27 &  30481 & 0 & 0 &0 &0\\
\hline
28 &  11 & 0   & 0 & 0  &0\\
\hline
29 &  0 &  0 &0  &  0  &0\\

\hline
\end{tabular}

\end{center}
\end{table}



\addcontentsline {toc}{section}{References}

\begin{thebibliography}{99}


\bibitem{Barg}
A. Barg, Complexity Issues in Coding Theory.  Handbook of Coding
Theory, V. S. Pless and W. C. Huffman, Eds., Elsevier, Amsterdam,
1998.

\bibitem{Bo_SJC07}I. Bouyukliev, What is Q-extension?, \emph{Serdica Journal of Computing} 1 (2007)
pp. 115-130.

\bibitem{BouyuklievJacobsson09}
I.~Bouyukliev and E.~Jacobsson, {Minimum lengths for codes with
given minimal primal and dual distance}, 2009.


\bibitem{BJV_IEEE00}I. Bouyukliev, D. Jaffe, and V. Vavrek, The Smallest Length of
Eight-Dimensional Binary Linear Codes with Prescribed Minimum
Distance, \emph{ IEEE Trans. Inform. Theory}, vol. 46, No.4,
(2000) pp. 1539-1544.

\bibitem{BV_IEEE05}I. Bouyukliev and Z. Varbanov, Some results for linear
binary codes with minimum distance 5 and 6. \emph{ IEEE
Transactions on Information Theory} 51(12), (2005) pp. 4387-4391.

\bibitem{Brou98}
{\sc A. E. Brouwer}, {\em Bounds on the size of linear codes}, in
Handbook of Coding Theory, V. S. Pless and W. C. Huffman, eds.,
Elsevier, Amsterdam, 1998, pp. 295--461.

\bibitem{Jaff97}  D. B. Jaffe,  Binary linear codes: new results on
nonexistence,  {Draft version accessible through the author's web
page http://www.math.unl.edu/} {djaffe}, April 14,1997 (Version
0.4) Dept. of Math. and Statistics, University of Nebraska,
Lincoln.

\bibitem{Jaffe}D. B. Jaffe,
Optimal binary linear codes of length $\le 30$, \emph{Discrete
Math.} {\bf 223} (2000) pp. 135--155.

\bibitem{new-Tables}
M. Grassl, Code Tables: Bounds on the parameters of various types
of codes, {\tt http://www.codetables.de/}.

\bibitem{HirSt}
J.W.P. Hirschfeld, L. Storme, The packing problem in statistics,
coding theory, and finite projective spaces, J. Statist. Planning
Inference 72 (1998) 355-380.

\bibitem{Kurosawa} K. Kurosawa and T. Satoh, Design of SAC/PC(l)
of order k boolean functions and three other cryptographic
criteria, in \emph{Advances in Cryptology – EUROCRYPTO'97}, ser.
Lecture Notes in Computer Science, vol. 1233. Springer-Verlag,
1997, 434--449.


\bibitem{Matsumoto}
R. Matsumoto, K. Kurosawa, T. Itoh, T. Konno and T. Uyematsu,
Primal-dual distance bounds of linear codes with application to
cryptography, \emph{IEEE Trans. Inform. Theory} 52, 2006,
4251-4256.

\bibitem{Ostergard-clique}
P.R.J. \"Osterg{\aa}rd, A fast algorithm for the maximum clique
problem, \emph{Discrete Applied Mathematics} 120 (2002), 197-207.



\bibitem{Hand} V.Pless, W.C.Huffman, R.Brualdi,
 An Introduction to Algebraic Codes,
in \emph{Handbook of Coding Theory}, V. S. Pless and W. C.
Huffman, Eds., Elsevier, Amsterdam, 1998, pp. 177--294.

\end{thebibliography}

\begin{table}

{\small
Algorithm \textsc{Extend} (Input: $V$, $t$):\\

\noindent
$|$~~~ \textbf{var} U: array of sets; \  max,i:integer;\\
$|$~~~~~~~~~~  found:boolean;  $\hat{X}$:set;
  r:array of integer;\\
\textbf{Procedure}  ext($\hat{X}$:set; size:integer);\\
$|$~ \textbf{var} i:integer;\\
$|$~ \textbf{ \{}\\
$|$~~~~~\textbf{if} size= $t$ \textbf{then}   \textbf{ \{} print($\hat{X}$); \textbf{exit}; \textbf{ \}};\\
$|$~~~~~\textbf{if}  $|U{[size-1]}|=0$ \textbf{then}\\
$|$~~~~~~ \textbf{ \{}\\
$|$~~~~~~~~~~~\textbf{if} ((size$>$max) and (size$<t$)) \textbf{then}\\
$|$~~~~~~~~~~~~\textbf{ \{}\\
$|$~~~~~~~~~~~~~~~~~~~~max:=size;\\
$|$~~~~~~~~~~~~~~~~~~~~\textbf{if} $max<t$ \textbf{then} found:=true;\\
$|$~~~~~~~~~~~~\textbf{ \}};\\
$|$~~~~~~~~~~~~~exit;\\
$|$~~~~~~~\textbf{ \}};\\
$|$~~~~~~~\textbf{while} $|U{[size-1]}| <> 0$ \textbf{do}\\
$|$~~~~~~~\textbf{ \{}\\
$|$~~~~~~~~~~\textbf{if} $(((size + |U{[size-1]}|<=max)$ and $(max<t))$ or\\
$|$~~~~~~~~~~~~$((size + |U{[size-1]}|)<t)$ and $(max=t)))$ \textbf{then} exit;\\
$|$~~~~~~~~~~$i:=\min\{j: v_j \in U{[size-1]}\};$\\
$|$~~~~~~~~~~$\textbf{if} \ (((size + r[i])<=max)$ and $(max<t))$ or\\
$|$*~~~~~~~~~~~~~~~$((size + r[i])<t)$ and $(max=t)))$  \textbf{then}   exit;\\
$|$~~~~~~~~~~$U{[size-1]}= U{[size-1]} \backslash \{v_i \};$\\
$|$**~~~~~~~~$U{[size]}:=\{v_j: v_j \in U{[size-1]}$ and $\{\hat{X}\cup v_i \cup v_j\} \in P \}$;\\
$|$~~~~~~~~~~\textbf{if} $size< t$ \textbf{then} ext($\{\hat{X}\cup v_i \}, size + 1$);\\
$|$~~~~~~~~~~\textbf{if} $found=true$ \textbf{then} exit;\\
$|$~~~~~~~\textbf{ \}};\\
$|$~\textbf{ \}};\\

\noindent
\textbf{Procedure} Main;\\
$|$ \textbf{\{}\\
$|$~~~max:=0;\\
$|$~~~\textbf{for} $i:=|V|$ \textbf{downto} 1 \textbf{do}\\
$|$~~~\textbf{\{}\\
$|$~~~~~~~~found:=false;\\
$|$~~~~~~~~$\hat{X}:=\{v_i \}$;\\
$|$~~~~~~~~$U{[0]}:=\{v_j: v_j \in V \,  \ j>i$ and  $\{\hat{X} \cup v_j \}\in P$ \};    \\
$|$~~~~~~~~ext($\hat{X}$, 1);\\
$|$***~~~~~r[i]:=max;\\
$|$~~~\textbf{\}};\\
$|$~\textbf{\}};\\
}

\end{table}

\end{document}